\documentclass[runningheads]{llncs}

\usepackage{graphicx}

\title{Visibility Extension via Reflective Edges to an Exact Quantity}
\author{Arash Vaezi \thanks{Department of Computer Engineering, Sharif University of
Technology, {\tt avaezi@ce.sharif.edu}}
\and
Mohammad Ghodsi\thanks{Department of Computer Engineering, Sharif University of
Technology, and Institute for Research in Fundamental Sciences
(IPM), Tehran, Iran. {\tt ghodsi@sharif.edu}. This author's research
was partially supported by the IPM under grant No: CS1392-2-01}}

\def\seg#1{\overline{#1}}
\def\P{\cal P}
\begin{document}
\thispagestyle{empty}
\maketitle
\begin{abstract}
We consider extending the visibility polygon $(VP)$ of a given point $q$ $(VP(q))$, inside a simple polygon $\P$ by converting some edges of $\P$ to mirrors. We will show that several variations of the problem of finding mirror-edges to add precisely $k$ units of area to $VP(q)$ are NP-complete. The optimal cases are NP-hard. We are unaware of any result on adding an \emph{exact} number to a polygon, or covering an area with an \emph{exact} surface. We deal with both single and multiple reflecting mirrors for both specular or diffuse types of reflections.

\end{abstract}

\section{Introduction}
The visibility problem has undergone an extensive analysis in the literature. Linear-time algorithms have been proposed to obtain the visibility polygon(region) of a point $q$ within a simple polygon $VP(q)$~\cite{1}, or when the viewer is a segment \cite{2}.

If some of the edges of $\P$ are made into mirrors, then each $VP(q)$ may enlarge.
Visibility in the presence of mirrors was first introduced by Klee in 1969 \cite{3}. He asked whether every polygon whose edges are all mirrors is illuminable from every interior point. In 1995 Tokarsky constructed an all-mirror polygon inside which there exist a dark point~\cite{4}.

It was shown in 2010 that $VP$ of a given point or segment can be computed in the presence of a mirror in $O(n)$ time~\cite{6}, with respect to the complexity of $\P$.

Visibility with mirrors subject to different types of reflection also has been studied before \cite{5}. There are two reflection effects for the mirror-edges that we consider in this paper, diffuse-reflection to reflect light with all possible angles from a given surface and, the specular reflection which is the mirror-like reflection of light from a surface. In specular reflection, a single incoming direction is reflected into a single outgoing direction. Some have also specified the maximum number of allowed reflections via mirrors in between \cite{ad}.

\subsection{Our Results}
We initiate the study of an optimization problem in a setting related to the art gallery problem. In this setting, the polygon edges can be converted to mirrors(reflectors). A guard can see a point if it is directly visible to it or if it is mirror-visible via one or more reflections. This is a natural and non-trivial extension of the classical art gallery setting. This setting has been considered before in few papers \cite{walcom,euro,ad,5}. The problem considered in this paper is to find a minimum number of edges that can be converted to mirrors so that the visibility region of a given point gets expanded by an area of $k$. Depending on the reflection type (specular or diffuse) and number (single or multiple) one can consider four cases. In this paper, we have proved NP-completeness of the problem for all the cases except the specular multiple reflection case.

There are many variations of this theme.  Reflections from mirrored edges may be required to obey Snell’s law ("specular mirrors"), or they may simply connect any pair of rays incident to the same point on the mirror ("diffuse mirrors").  Rays may be permitted to bounce a limited number of times or travel through a limited number of edges. The viewpoint $q$ may become a segment inside $\P$ \cite{walcom}.

This article considers the following problem: given a polygon $\P$, a point $q$ inside $\P$, and positive integer $k$, does there exist a subset of edges of $\P$ that can be turned into mirrors so that $VP(q)$ increases in an area by exactly $k$ units?  Two problems are considered, one with specular mirrors and one with diffuse mirrors.

We present polynomial-time reductions from the NP-complete Subset-Sum problem to our problem and show that each solution for an arbitrary instance of the problem is verifiable for correctness in polynomial time. Thus, the primary NP-completeness results follow. We are working on another paper regarding optimization problems with monotone properties \cite{atleast}.

We dealt with another version of the problem of extending visibility polygons in another paper \cite{euro} and we improved our results in \cite{walcom}. In those papers, we had a given object (a point or a segment) as a target inside a simple polygon $\P$. This target is not visible to the given viewer (a point). In order to make the target mirror-visible to the viewer, we converted some edges of $\P$ to mirrors. Different types of segment visibility were considered. We proposed a linear time algorithm which finds all mirror-edges that make the target mirror-visible to the viewer, and reveals precisely which part of the target is mirror-visible through each mirror-edge.

In this paper, we do not have a target. To add an exact amount to $VP$ of a given viewer, the problem is more intricate and tricky. 

\section{Applications}
The art gallery problem originates from a real-world problem of guarding an art gallery with the minimum number of guards who together can observe the whole gallery. However, in some cases, it is not possible to use guards just anywhere in the gallery. There might be security reasons, or the building might be too old to install many guards where ever we need.

Consider a market where the owner cannot afford many cameras to guard the whole place. A simple way is to buy a high-resolution camera and to install it somewhere safe and accessible. And then, the owner can install some mirrors around the market. The minimum number of mirrors that can guard the whole place will minimize the overall cost. In this example, in order to have a better video quality, it is better to allow only one reflection for the mirrors. To keep the beauty of the place the mirrors may be installed on the wall of the store. The owner needs to find all the walls that a mirror should be installed on them, the location of each mirror on its corresponding wall, and the exact size of the mirror. 

To present an ideal example, consider a subway or a studio under the ground, assume that the telecommunication signals can only be received from an outside source by an antenna in the entry of the subway. The telecommunication signals should be transmitted from one place to another started from the entry. Some reflectors should amplify these signals and reflect them in various directions to ensure that the whole place is covered successfully. So, one can be sure that his/her cell phone works under the ground effectively. These reflectors are not ordinarily capable of reflecting the signals in all directions. 

A base transceiver station (BTS) is a piece of equipment that facilitates wireless communication between user equipment (UE) and a network. UEs are devices like mobile phones, computers with wireless internet connectivity. Though the term BTS can apply to any of the wireless communication standards, it is generally associated with mobile communication technologies. An antenna is a structure that the BTS lies underneath; it can be installed as it is or disguised in some way. The signal quality of an antenna is much better in one specific direction around the antenna. The minimum number of these antennas and their places is extremely important.

Here, we simplify the problem. There is a simple polygon $\P$, and a viewer $q$. The viewer has a visibility polygon and is capable of seeing in all directions. Also, there is no distance limitation for the viewer to see anything. We intend to find the minimum number of reflectors. These reflectors have no distance limitation too. They can reflect signals (or lights) coming from the viewer in a specific direction considering the specular type of reflection, or in all directions considering the diffuse type of reflection. These reflectors are installed on the edges of the polygon, and we assumed that no mirror can be inside the polygon. In fact, we can assume that we convert some part of an edge to a reflector. The problem is to find the place of the minimum number of reflectors to cover either the whole polygon or some pre-specified regions inside $\P$.

\section{Notations}
Suppose $\P$ is a simple polygon where $int(\P)$ denotes its interior. Two points $x$ and $y$ are visible to each other, if and only if the relatively open line segment $\overline{xy}$ lies completely in $int(\P)$.
The visibility polygon of a point $q$ in $\P$ denoted as $VP(q)$, consists of all points of $\P$ visible to $q$.

Every edge of $\P$ has the potential of converting into a mirror. We can assume that all edges are mirrors. However, the viewer can only see some edges of $\P$. From now on, when we talk about an edge, and we need to consider it as a mirror in order to compute or check something (such as its mirror-visibility area), we call it \emph{mirror-edge}. Two points $x$ and $y$ inside $\P$ can see each other through $e$, if and only if they are directly visible with a kind of reflection. 

Since only an interval of a mirror-edge is useful, we can consider the whole edge as a mirror, and there is no need to split an edge. Since $\P$ is a simple polygon, the viewer inside $\P$ can only see a contiguous (possibly empty) portion of any edge in $\P$. In other words, only one contiguous part of every edge of $\P$ is visible for the viewer. We need to find this part, and also make sure that it is visible to the target segment too. 

Two points or segments are "mirror-visible" if and only if they can see each other through a mirror-edge. 

In our reductions we will use the variation of the \emph{Subset-Sum problem} in which the target($T$) and all the values are non-negative, this variation is NP-Complete due to \cite{CLRS}.

\section{Expanding \emph{exactly} $k$ units}
We begin this section by the following theorem which is in fact the main contribution of this paper. 
\begin{theorem}
Given a simple polygon $\P$, $q \in \P$, and an integer $k >0$,
the problem of choosing $m$ mirror-edges of $\P$ in order to expand $VP(q)$ \emph{exactly} $k$ units of area is NP-complete in the following cases:

1. Specular-reflection type regarding single reflections.
 
2. Diffuse-reflection type regarding either single or finite-multiple reflections.

\end{theorem}
Considering the assumption that the above theorem is true, 
on the subject of the minimum $m$, all of the cases mentioned above are NP-hard (The optimal case of the problem is NP-hard).

By checking each edge mirror-visibility, in polynomial time, a given solution can easily be verified if it adds precisely $k$ units to $VP(q)$. Therefore, the problem is in NP. According to \cite{6} we can check each edge mirror-visibility in linear time, and we can compute the area of the union of all mirror-visibility areas in polynomial time.

Note that, the overall structure of the polygons we construct in our reductions can be created in a way that all the coordinates become a rational number. That is because we can easily re-scale the polygon so that all the angles remain constant, or also, we can shift the gadgets in the polygons without any damage to the construction of the reduction.

We will show that the NP-complete Subset-Sum problem is reducible to this problem in polynomial time. Thus, we can deduce that our problem in the cases mentioned above is NP-complete. 

We will discuss the details of our reductions in the following subsections. 

\subsection{Specular type of reflections}

Consider an instance of a Subset-Sum problem ($InSS$), which has $v_{1} ,v_{2} ... ,v_{n}$ non-negative integer values, and a target number $T$. The Subset-Sum problem is to look for a subset of these values, which their summation equals to $T$.

The main contribution of the reduction is the structure of the polygon, which provides some features for the polygon.  We use a simple polygon (similar to the one in Fig.~\ref{fig.1}) where there is only one edge of the constructed polygon
that allows $q$ to completely see one specific quadrangular spike opposite to its bottom edge. Note that every edge of the polygon has the potential of getting converted to a mirror-edge but for every spike in the polygon (illustrated in \ref{fig.1}) there is only one edge which can make that spike mirror-visible. In this Subsection, we only consider the specular type of reflection. And, spikes and their corresponding mirror-edges are in the same order from right to left. Moreover, no other edge of $\P$ can see even a $\epsilon$ area of a spike except for its corresponding mirror-edge in the order from right to left.  Each area of those spikes equals to a value in $InSS$. 

For obtaining exactly $k$ units of additional area, some mirror-edges facing their corresponding quadrangular spikes need to be selected and converted to mirrors. The set corresponding to quadrangular spikes seen by a set of selected opposing edges to be converted to mirrors in order to add exactly $k$ units of area to the visibility polygon of $q$ will correspond to a solution for the specified instance of the Subset-Sum problem, and the reduction is complete.

In this reduction, each quadrangular spike is totally and only visible through one mirror-edge, which is why this reduction and construction does not work in the case of multiple reflections (Two rays demonstrated in \ref{fig.1} within the polygon reveals that this reduction does not work in the case of multiple reflections). Even if tiny little units of an area from different spikes were partially visible through some other potential mirror-edges (except for their corresponding mirror-edges), their summation could be a large number and the reduction might fail. Please see appendix for more details.

To be more precise on the structure of the polygon $\P$, see Fig.~\ref{fig.1.seperated}, the reduction polygon consists of two main components. The first component illustrated in Fig.~\ref{fig.1.seperated}(A) is a vertical rectangular. $b$, $c$, $d$ and $e$ are the name of four vertices of this rectangular. Two other vertices of this part of the polygon are $a$ and $f$, which these vertices are on the unique line which contains segment $\seg{be}$. Between $a$ and $f$, there is no line.

We choose integer coordinates for all the mentioned above vertices of $\P$.

 If we imagine a line parallel to the segment $\seg{ed}$, which starts from $f$ inside the polygon, $q$ (the viewer) should be on this line. The more $q$ is close to $f$ the smaller the reduction polygon will be. However, the position of $q$ on this imaginary segment has a considerable effect on the coordinates of the vertices of the other component of the polygon. 

Imagine another segment which crosses $q$ and parallel to $\seg{cd}$. The two imaginary segments, parallel to the edges of the rectangular, create four sub-rectangular spaces in the polygon. Fig.~\ref{fig.1.seperated}(A) reveals how the rays emitted from $q$, via one specular reflection, can be reflected in each of these sub-rectangular spaces. The main rectangular provides such kind of space that $q$ cannot see anything via one specular reflection except on the left side of $\seg{af}$ and via upward rays. 

Fig.~\ref{fig.1.seperated}(B) illustrates the other component of the reduction polygon. This component has $3n$ vertices on its top and $4n$ vertices on its bottom ($n$ is the number of the values of $InSS$). 

On the top of the polygon, there are $n$ mirror-edges. Each mirror-edge has the potential of adding an area to the visibility polygon of $q$. These mirror-edges are numbered from right to left. The vertices of every mirror-edge are on a unique line which  starts from $v_{1}(e_{1})$ (or $a$), and ends at $v_{2}(e_{n})$. $q$ can see every point on this line. So, we select $n$ integers to be the coordinates of the mirror-edges.

Every segment $\seg{v_{2}(e_{i})v_{3}(e_{i})}$ $1 \leq i \leq n$ lies on a ray emitted from $q$. And, every segment $\seg{v_{3}(e_{i})v_{2}(e_{i+1})}$ $1 \leq i \leq (n-1)$ is perpendicular to the next mirror-edge ($\seg{v_{1}(e_{i+1})v_{2}(e_{i+1})}$ $1 \leq i \leq (n-1)$).

On the bottom of the component shown in Fig.~\ref{fig.1.seperated}(B), there is a line starting from $f$ and parallel to the line containing the mirror-edges on the top of this component. The underneath line contains $4n+1$  vertices, and there are $n$ windows. We call this line; the window-line. The windows on the window-line are numbered from right to left. Every window $w_{i}$ corresponds to the mirror-edge $e_{i}$, $1 \leq i \leq n$. The ray reflected from $e_{i}$ on $v_{1}(e_{i})$ reaches the window-line on a point denoted as $v'_{1}(e_{i})$. And, the ray reflected on $v_{2}(e_{i})$ intersects the window-line on a point denoted as $v'_{2}(e_{i})$ $1 \leq i \leq n$. The window-line starts at $f$ and ends at $v'_{2}(e_{n})$.

Consider a segment $\seg{v'_{1}(e)v'_{2}(e)}$, if the coordinates of either $v'_{1}(e)$ or $v'_{2}(e)$ is not rational, then we lessen the size of the corresponding window and choose the first rational number on that window. For example, if the coordinates of $v'_{1}(e_{1})$ was not a rational number then we move to the left on the window-line and select a point with rational coordinates. Note that as we select integer coordinates for the vertices of every mirror-edge, the size of each window is more than one unit (the space between two consecutive integers). So, there is enough space to move $v'_{1}$ or $v'_{2}$ inside the $\seg{v'_{1}v'_{2}}$ interval on window. Also, there is enough space between two windows, and no mirror-edge can disturb the functionality of another one. That is the region behind each window is only and entirely visible to its corresponding mirror-edge. 

By joining the two aforementioned components, we can construct a simple polygon which is, in fact, the reduction polygon.

See Fig.~\ref{fig.1}, $q$ cannot see any invisible region through $\overline{ab}$, $\overline{bc}$, $\overline{cd}$, $\overline{de}$, and $\overline{ef}$ edges.

\begin{figure}[htb]
\begin{center}
\includegraphics[width=\textwidth]{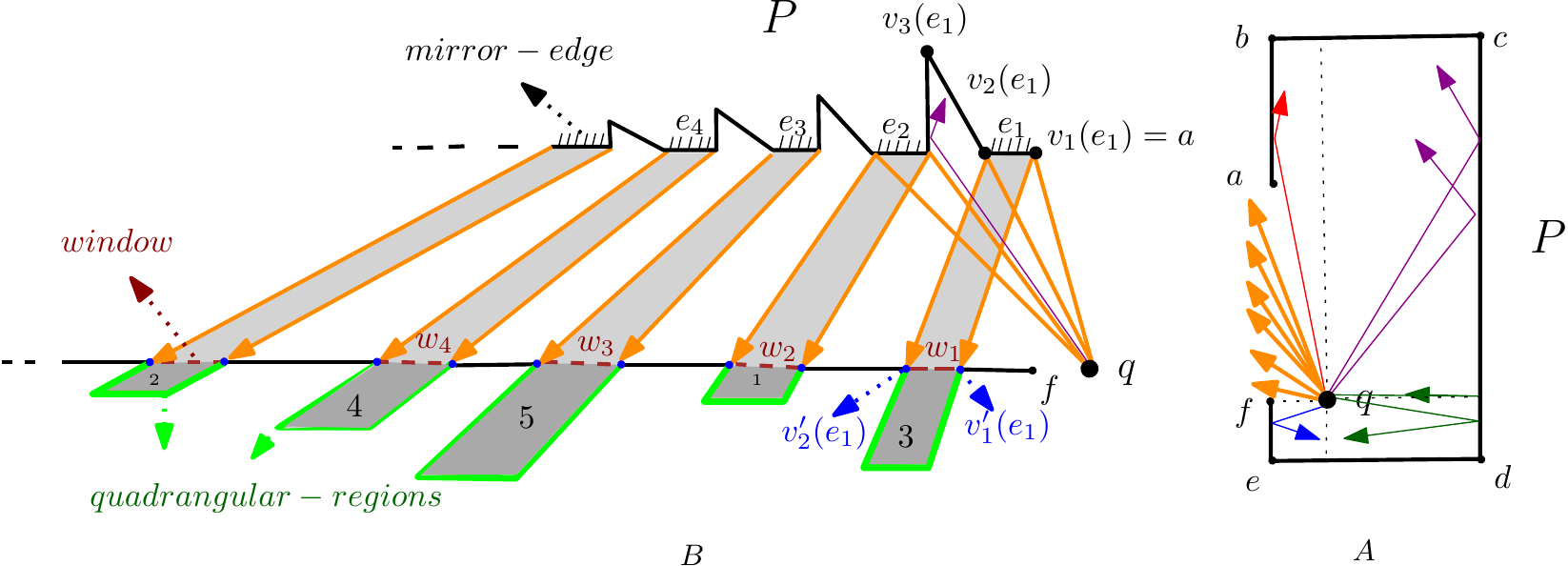}
\caption{Two main components of the reduction polygon is illustrated.} 
\label{fig.1.seperated}
\end{center}
\end{figure}

\begin{figure}[htb]
\begin{center}
\includegraphics[width=\textwidth]{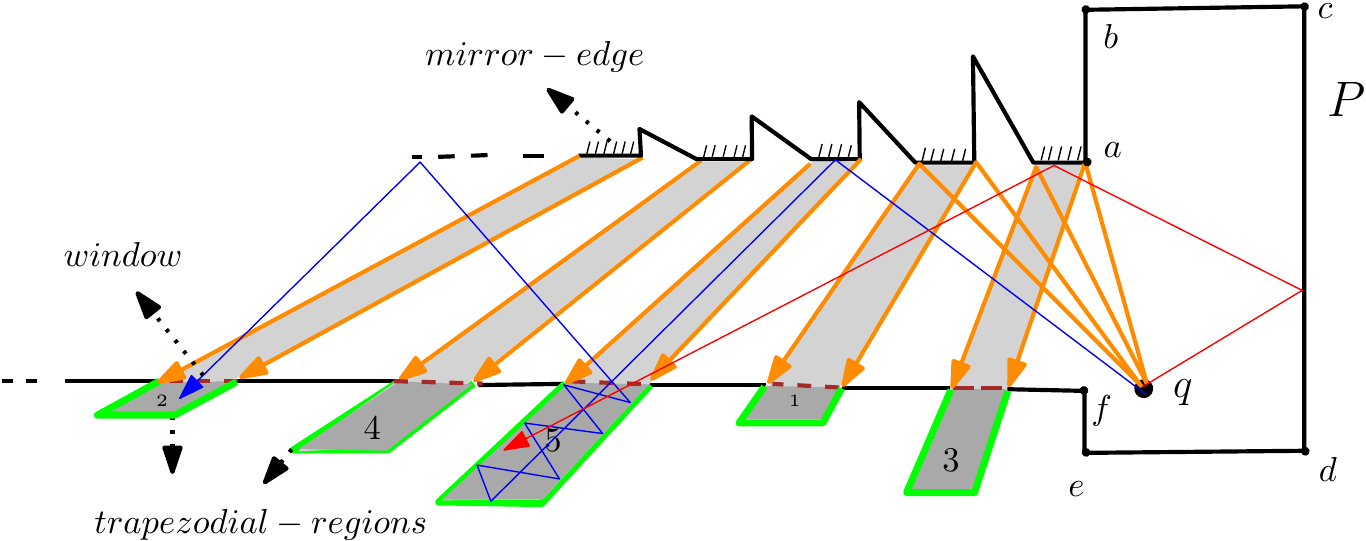}
\caption{Reduction Polygon. There is a quadrangular region behind each window. And the surface area of this region corresponds to a value of $InSS$. This figure contains five first values 3, 1, 5, 4, 2. } 
\label{fig.1}
\end{center}
\end{figure}

Finally, we need to 
put a quadrangular region with an area equal to $v_{i}$ units of area, behind the $ith$ (for $ 1 \leq i \leq n $) window so that $q$ can see that region entirely through the corresponding mirror-edge, and no other edge can make this region visible even a little bit.

To set up the quadrangular region behind a window $w$, after computing the correct positions for the $v'_{1}(e)$ and $v'_{2}(e)$ corresponding vertices, we find the reflected rays form $e$ on $v'_{1}(e)$ and $v'_{2}(e)$.  We need to extend these rays to provide a quadrangular region. Every quadrangular region has four vertices; $v'_{1}$ and $v'_{2}$ and two other on a segment lower than the window-line. This lower segment of each quadrangular region is parallel to the window-line. Since the height of each quadrangular region is flexible, we can set any surface area for that region.

Note that both base and height of a quadrangular region are variable.

Furthermore, the reduction is from Subset-Sum problem, which is weakly NP-hard. Therefore, the integer coordinates should be polynomially bounded in the values and $n$. 

\subsubsection{Multiple Reflections}
The multiple reflection case of this version of the problem is still open.

Fig.~\ref{fig.1.seperated}(B) illustrates a ray reflected from $\seg{v_{3}(e_{1})v_{2}(e_{2})}$, after finite number of reflections this ray cannot get back to the bottom of the polygon.

However, as Figure \ref{fig.1} shows the multiple reflection rays may disturb the functionality of the mirror-edges. Even after two reflections the segment $\seg{cd}$ may cause a mirror-edge to see some area behind another window. Also, this mirror-edge can see the area behind its own corresponding window. 

Considering multiple reflections, on the left side of the polygon the segment $\seg{v_{2}(e_{n})v'_{2}(e_{n})}$ may disarrange the functionality of some mirror-edge, too.

\subsection{Diffuse type of reflections}
This version of the problem is more intricate. And, we cannot use the same polygonal structure used for the previous case. The reason is that in diffuse type of reflections, rays are reflected to the all directions and can be reflected into wrong spikes (a spike which should get mirror-visible via another mirror-edge). Considering multiple plausible reflections the problem gets even harder. These rays have to be excluded by an appropriate arrangement of the polygon's edges. In this subsection, We will present a completely different construction in compare with the construction created in the previous subsection. The construction presented in this subsection works in the case of multiple reflections, too. However, again we reduce the Subset-Sum problem to our problem.

The construction is as the following:

Consider the single reflection first. The reduction polygon consists of some gadgets. We have one gadget for each value in the given instance of the Subset-Sum problem (denoted as $InSS$). So, we need to set $n$ gadgets. Each gadget has a window. And, $q$ (the viewer) can not see behind that window directly. 

Every gadget has a specific edge that can make the area behind the window entirely mirror-visible to $q$ via diffuse type of reflection; we call this edge as the mirror-edge (or later the main-mirror-edge). Other edges of a gadget should not see any area behind the window via a single diffuse type of reflection.

\begin{figure}[htbp]
\begin{center}
\includegraphics[scale=1.3]{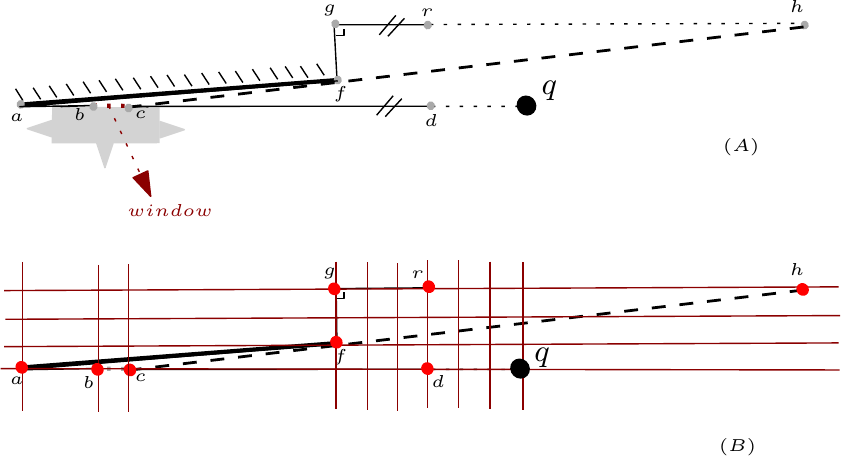}
\caption{Gadget for each value in the given Subset-Sum instance. (A) Illustrates the gadget structure. The gray region reveals the wide mirror-visible area behind the window. This area can become mirror-visible through the mirror-edge via diffuse type of reflection. (B) Shows an example that reveals that the coordinates of the vertices of a gadget can be rational.} 
\label{fig.2}
\end{center}
\end{figure}

See Fig.~\ref{fig.2}.(A). In this gadget, $a$, $b$, $c$, $d$ vertices, and the given point $q$ lie on one unique line. $h$, $f$, and $c$ vertices lie on another line. Also, $h$, $r$, and $g$ vertices lie on a parallel line with $\overline{cd}$. $a$ is connected to $f$. In each gadget, $\overline{af}$ can play the same role as a mirror-edge in the previous reduction presented in the previous subsection. In other words, consider a gadget in the final polygon, and suppose a solution chooses to convert the $\overline{af}$ edge to a mirror. Then, there will be one specific spike (area) which is only and entirely mirror-visible through this $\overline{af}$ edge. And this area should be behind a window segment. 

$\seg{bc}$ is the window in each gadget. The arrangement of other edges of a gadget is in a way that no edge except for $\seg{af}$ can make any area behind $\seg{bc}$ mirror-visible to $q$.

Since $\overline{gr}$ should not make any part of the gray region mirror-visible, $r$ should be on the left side of $h$. This way, $c$ and $f$ vertices block $\overline{gr}$'s mirror-visibility for $q$ not to see the gray region behind the window.

\begin{lemma}
The coordinates of a gadget can be computed in polynomial time in a Turing machine.
\label{lem.1}
\end{lemma}

\begin{proof}

Consider one gadget like the one illustrated in Fig.~\ref{fig.2}. First, we select integer coordinates for all the vertices. Fig.~\ref{fig.2}(B) exhibits an example. We will see that the vertical and horizontal distances between the vertices of a gadget, and also between $q$ and these vertices are flexible. 

Since we consider the diffuse reflection type, the size of the mirror-edge is flexible, and $a$ can move to be close enough to $b$ by changing the size of the mirror-edge. And, these movements will not change the visibility of the mirror-edge on the other side of the window.

In particular, consider the following examples;
\begin{enumerate}
    \item If we move $f$ to a lower position with rational coordinates $h$ will move to its right side. 
    \item If we move $g$ to a lower position, $h$ will move to a lower position too.
    \item We can change the position of $c$ to the left or right, and $h$ will move accordingly.
    \item $b$ has a more flexible position. If we fix the position of $c$, by moving $b$ to the left or right we can change the size of the window. Also, The position of $a$, $r$, and $d$ vertices is highly flexible too.  
\end{enumerate}

Now suppose that we construct $n$ gadgets with rational coordinates independently. In the final structure of the polygon, we need to put these gadgets in different positions and connect them. We will see how to do this after Lemma~\ref{lem.1} in the rest of the paper. However, we may rotate or move a gadget to another place, and the coordinates of its vertices will change accordingly. By keeping the overall arrangement of a gadget, we can change the positions of its vertices to have rational coordinates. 

We can change the position of $b$, $c$ and $f$ to provide a constant condition for the area behind the window.

Here, we count on $h$ as a vertex of the gadget, but later in the final structure of the polygon, we will replace it with a $r$ vertex of another gadget. As the coordinates of all the vertices of all gadgets are rational, this change will keep all the coordinates rational.
\end{proof}

See Fig.~\ref{fig.3}, if any edge of the final polygon stands in front of the dark gray area (the dark gray triangle which has $h$, $b$, $c$, $d$, $q$ on its boundary), then it has the potential of making a part of the area behind the window mirror-visible to $q$. This part is unavoidable to be directly visible through an edge rather than the specified mirror-edge of the gadget. We denote this small mirror-visible area as $SMV$. In Fig.~\ref{fig.3} this region in shown in green.

The reason that the $SMV$ area is unavoidable to get mirror-visible by another edge rather than the predetermined mirror-edge is this: There must be at least one edge (say $\overline{rd}$) in front of the $SMV$ area (or in front of the dark gray area shown in Fig.~\ref{fig.3}) for the final constructed polygon to be a connected and simple polygon. 

For every gadget, we set one $\overline{rd}$ edge on the opposite side of the polygon that can make the \emph{entire} $SMV$ area mirror-visible to $q$.

\begin{figure}[htbp]
\begin{center}
\includegraphics[scale=1.3]{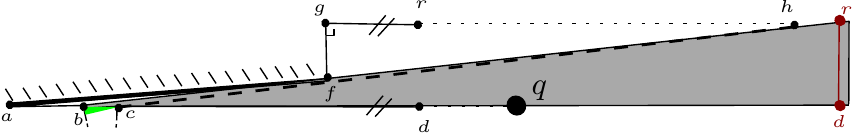}
\caption{The $SMV$ area is illustrated. Any edge of the main polygon in the dark gray area can see at least a part of the little green area.} 
\label{fig.3}
\end{center}
\end{figure}

In the worst case, base on the position of the $c$ and $f$ vertices, the surface area of this mirror-visible part is restricted to be the little $SMV$ area. This area is illustrated clearly in Fig.~\ref{fig.5}. The mirror-visibility of the $rd$ edge shown in Fig.~\ref{fig.5} covers the $SMV$ area completely.

\begin{figure}[htp]
\begin{center}
\includegraphics[scale=1.3]{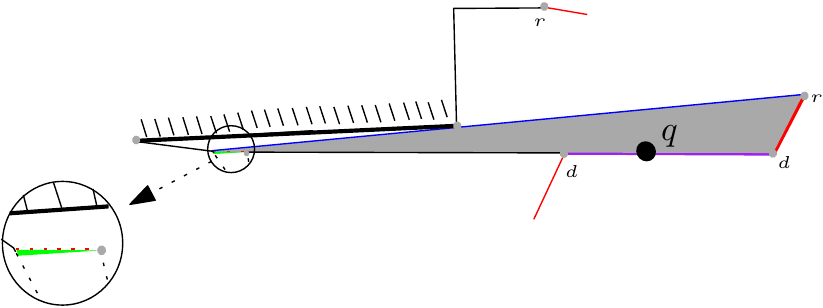}
\caption{The little $SMV$ region is shown in green. Every $rd$-edge can make an entire and only one $SMV$ region mirror-visible to $q$.} 
\label{fig.5}
\end{center}
\end{figure}

\begin{lemma}
The size of the surface of the $SMV$ area is flexible.
\label{lem.2}
\end{lemma}

\begin{proof}
As we mentioned earlier, the positions and sizes of vertices and segments are variable in many ways. In particular, considering a specific gadget, $b$ can move readily between $a$ and $c$, or $f$ can move to a lower position. Moreover, we can set the angle between $\overline{af}$ and $\overline{ab}$ accordingly. Therefore, we can set the surface of the $SMV$ area to cover any small predetermined units of area. 
\end{proof}

\subsubsection{The Two-Gadget Structure}

The dark gray region is between two segments; $\overline{bf}$ and $\overline{cd}$. To make sure no gadget disturbs another one, we need to know that no part of a gadget is in front of the region between $\overline{bf}$ and $\overline{cd}$ of another gadget.

Suppose that the number of the values of $InSS$ is even. We place every two gadgets head-on to each other in a way that $a$, $b$, $c$, $d$ vertices of each gadget, and $q$, lie in one line. Also, in every such construction of two gadgets, one gadget is upside-down in compare with the other one (see Fig.~\ref{fig.6}). We call this structure a two-gadget structure. Based on this structure, no dark gray region of any two gadget faces any edge of another one internally (from the inside of the gadget). 
This arrangement (the two-gadget construction) helps us to use $rd$-edges practically. We will discuss these edges later.

\begin{figure}[htbp]
\begin{center}
\includegraphics[scale=1.3]{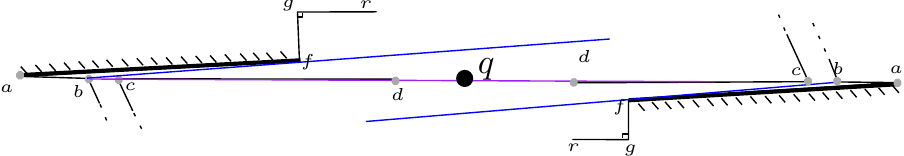}
\caption{We place every two gadgets head-on to each other in a way that the half-lines containing $\overline{dc}$ segments lie on one line.} 
\label{fig.6}
\end{center}
\end{figure}

\subsubsection{Assembling the Two-Gadget Structures}

All the two-gadget structures will be assembled in the way as Fig.~\ref{fig.7}(A) illustrates. The two half-lines containing the $\overline{bf}$ segments of every two-gadget should cross $r$ vertices of two other gadgets in the top and bottom of a two-gadget. 

Thus, no edge of any gadget can place in the dark gray area of other gadgets. Now, we add reliably all these two-gadget patterns one by one to the polygon. The gadgets will be around $q$ circularly. 

Note that, we can set various distances between the gadgets of the polygon by applying appropriate length to the segments $\overline{gr}$ and $\overline{fg}$. And also, the angle between each mirror-edge and its corresponding $\overline{ab}$ is adjustable.

\begin{figure}[htb]
\begin{center}
\includegraphics[width=\textwidth]{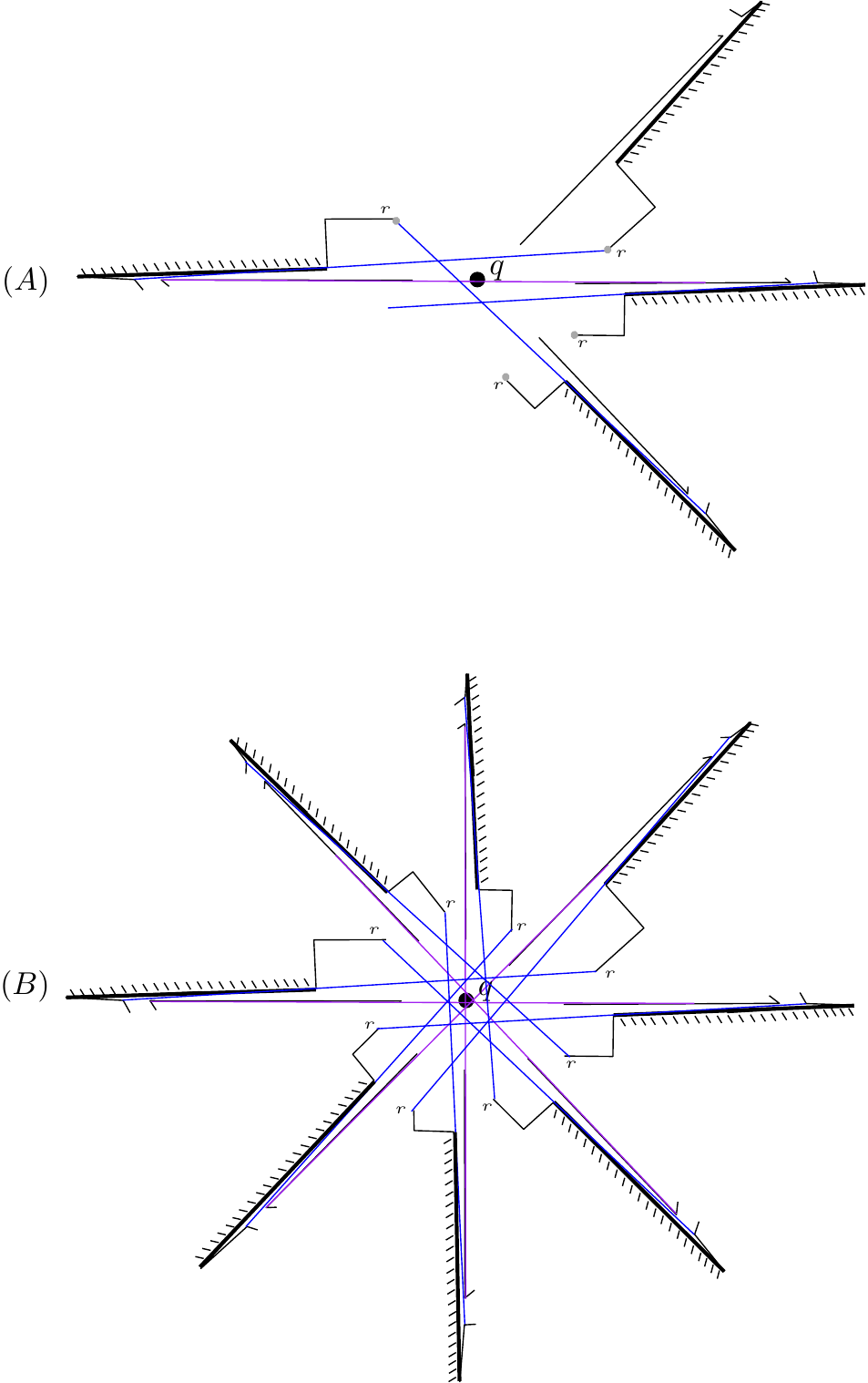}
\caption{(A) Shows how to attach gadgets to each other, and (B) Illustrates 8 gadgets. As we can see, the polygon is not completely constructed yet.} 
\label{fig.7}
\end{center}
\end{figure}

The segments $\overline{fg}$, $\overline{cd}$ and $\overline{gr}$ of a gadget are not in the dark gray region of another one. Therefore, these segments cannot disturb the functionality of the gadgets. In other words, in the constructed polygon up to this step, no edge except for the specified mirror-edge in each gadget can make any part of the region behind a window mirror-visible to $q$.

\subsubsection($\seg{rd}$-edges)
Fig.~\ref{fig.7}(B) demonstrates assembling of 8 gadgets in a polygonal space. The reduction polygon is not completely constructed. The gadgets are not connected. And, this is not a simple polygon yet. Any edge that fixes an unconstructed part can make one specific $SMV$ area on the opposite side completely mirror-visible to $q$. That is because any such edge will cover the entire dark gray area from the front side (the same as the $\seg{rd}$ edge illustrated in Fig.~\ref{fig.5}). So, the $SMV$ area will be covered entirely.

As Fig.~\ref{fig.8}, shows, we fill the unconstructed regions of the polygon by $\seg{rd}$-edges. Each $rd$-edge makes exclusively an entire $SMV$ area of another gadget mirror-visible to the viewer. This $SMV$ area corresponds to a pattern head-on the $\seg{rd}$-edge across $q$.

\begin{figure}[htp]
\begin{center}
\includegraphics[width=\textwidth]{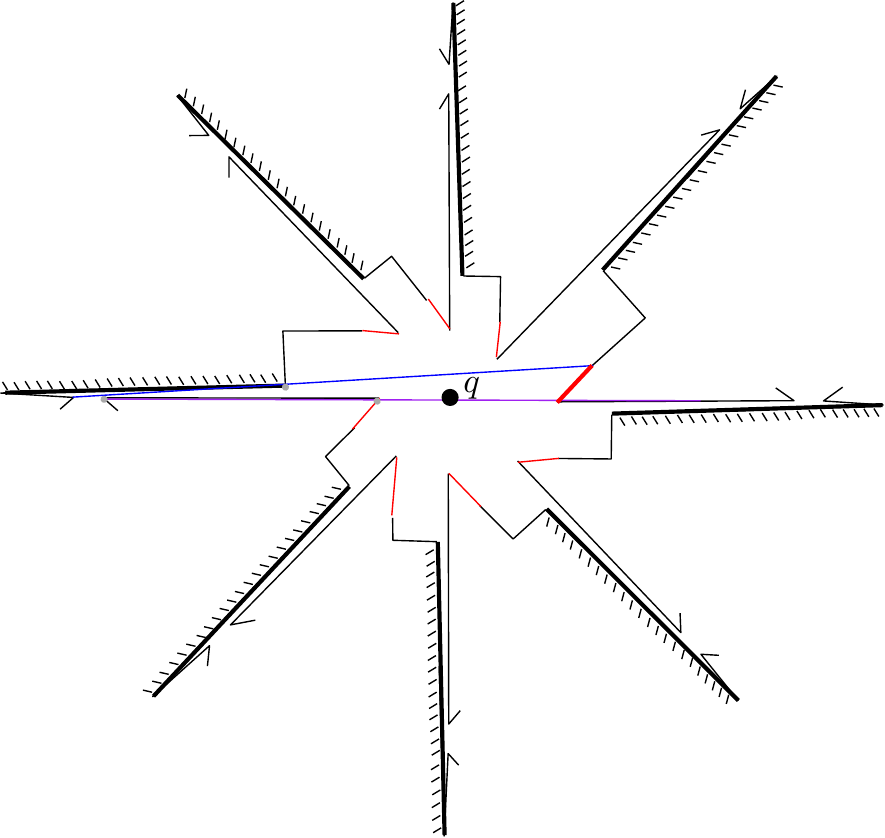}
\caption{The $rd$-edges connect the gadgets to each other. These edges can see a particular part of the region behind the windows (the little $SMV$ area).}
\label{fig.8}
\end{center}
\end{figure}

\subsubsection{The Area of the $SMV$ Regions}

The construction of the polygon and the gadgets should be in a way that, the summation of the surface area of all the $SMV$ regions gets equal to a value lower than 1 unit of area. This is practically feasible because we know the number of values in our given Subset-Sum instance is $n$, and also because the positions of $b$ and $c$ vertices of the gadgets are variable. So, we can adjust the surface area of every $SMV$ region. Assume that the summation of the surface area of all the $SMV$ regions is equal to $\epsilon$ (lower than 1 unit).

\subsubsection{Quadrangular Regions}

In the next step, we create $n$ quadrangular regions; each one has a surface area equal to one of the values of the given instance of the Subset-Sum problem ($InSS$). There might be not enough space in front of each window in the gadgets to put any quadrangular region with any shape. However, we already know that the type of reflections is diffuse, and the reflected rays from the mirror-edge are directed in all directions. So, behind a window, we have a wide range of space in different angles that can provide enough space for the quadrangle to be expanded to the outside of the polygon (see Fig.~\ref{fig.9}). In fact, this is the reason why we cannot use this reduction if we consider the specular type of reflections. 

To have a more precise discussion on the quadrangular region see the following:

A quadrangular region with a predetermined surface area can be placed behind $\overline{bc}$ (the window) of any gadget in a way that the whole invisible region of the quadrangular region gets visible to its corresponding mirror-edge. Behind a window, there is a quadrangular region whose surface area equals to a value of $InSS$ plus a $\epsilon$. This $\epsilon$ equals to the area of the corresponding $SMV$ area behind that window.
See Fig.~\ref{fig.9}(A). In this figure for the quadrangular regions to be clearer, we illustrate the $SMV$ areas in green. The gray-quadrangular regions are the quadrangular regions that their surface area equal to the values of $InSS$.

The $SMV$ area is a triangle with three vertices; $b$, $c$, and $x$. The position of $x$ is adjustable, but it provides different surfaces for the corresponding $SMV$ area. See Fig.~\ref{fig.9}(B). 

The gray-quadrangular area starts from $\seg{cx}$. Two other vertices on the other side can provide a quadrangle. The surface area of this gray-quadrangle is adjustable and flexible. The position of the four vertices of this quadrangle can be set to provide any surface area for it. Also, this flexibility ensures us that the coordinate of the vertices can be set to be rational.

\begin{figure}[htbp]
\begin{center}
\includegraphics[scale=1.3]{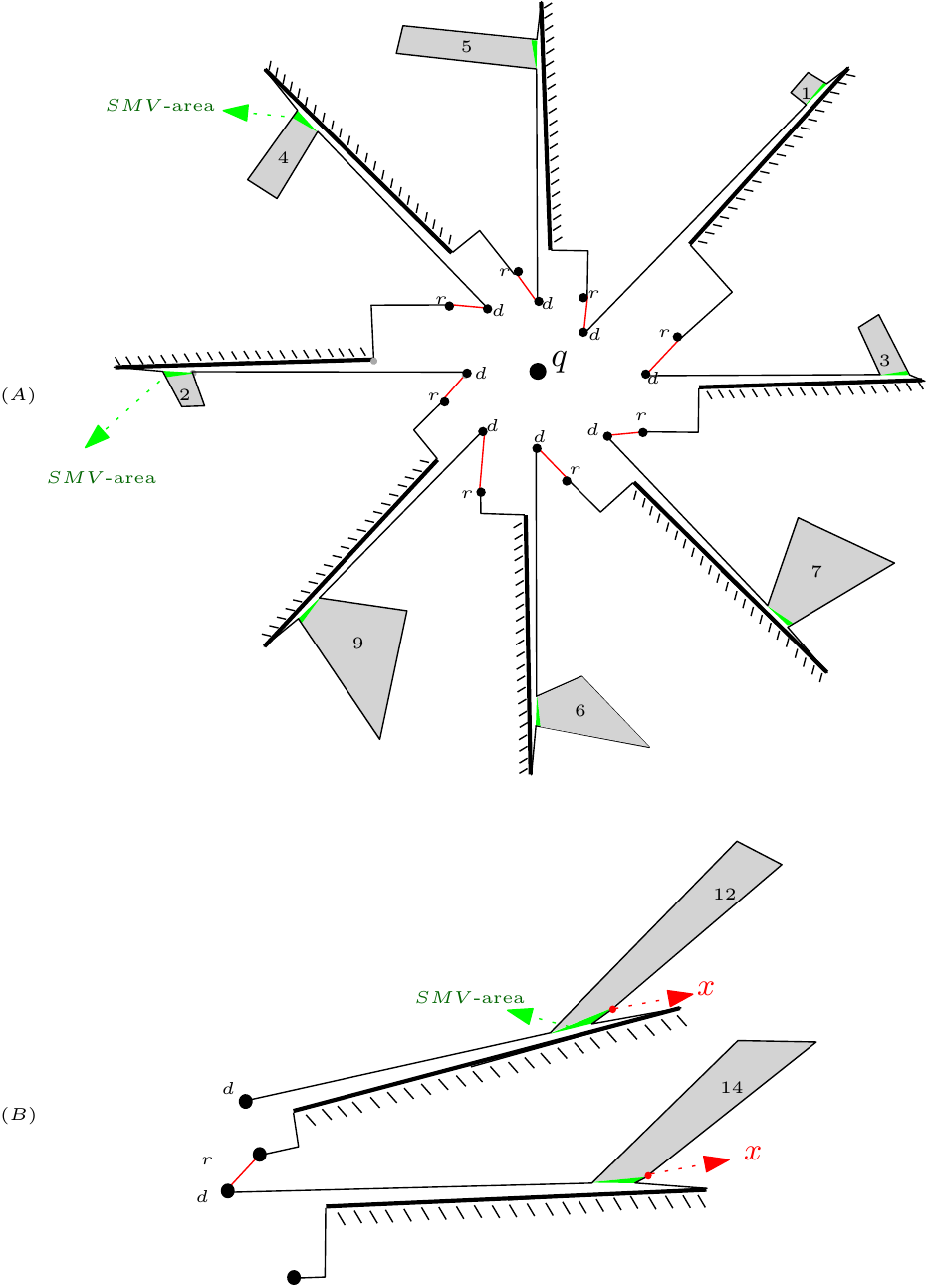}
\caption{(A) At the end, we add the gray-quadrangular regions to the polygon. (B)The gray-quadrangular region may expand to the outside of the polygon to provide any surface area which is required in a value of $InSS$. } 
\label{fig.9}
\end{center}
\end{figure}

\subsubsection{The $\epsilon$ Additional Space}

We set the target of our problem to $T+\epsilon$, which $T$ is the target value of $InSS$. And, $\epsilon$ is the summation of the surface areas of all the $SMV$ regions in all gadgets. So, We need all the $SMV$ areas to be mirror-visible.

Every $SMV$ area can be visible for $q$ through a $\seg{rd}$-edge or a mirror-edge ($\overline{af}$). When we select a mirror-edge to make visible the area behind its corresponding window, it is not required to choose the corresponding $\seg{rd}$-edge. However, in such a case that a mirror-edge is not selected, the corresponding $\seg{rd}$-edge has to be chosen so that the whole $\epsilon$ surface gets covered.

\subsubsection{Odd Number of Values}

If the given Subset-Sum instance ($InSS$) has an odd number of values, then we can easily add one artificial gadget to our structure. First, we generate a two-gadget structure. Then, we can put an extra edge between $r$ and $d$ vertices of one of its gadgets, and eliminate that one.

Note that according to our previous discussions all vertices have rational coordinates so eliminating a gadget in this way does not affect our reduction. 

Since this edge is not in the dark gray area of any other gadget (the area between the two lines that include $\seg{bf}$ and $\seg{cd}$ segments in a gadget), its mirror-visibility cannot disturb anything in our final constructed polygon. And, the reduction will work fine.


\subsubsection{Reduction Analysis}

There are $n$ values in the given instance of the Subset-Sum problem ($InSS$)
 
Each $SMV$ area is a triangle. So, we can compute the complete surface area of all the $SMV$ areas in $O(n)$ time.  Also, the surface area of each gray-quadrangular region (the area behind a window minus the $SMV$ area) can be set to be equivalent to a value in $InSS$. To do this, we use a quadrangular shape and set its base and height appropriately as we mentioned previously. We can extend the gray-quadrangle to the outside of the polygon.

We call the mirror-edges ($\overline{bc}$) used in the above reduction, as the main-mirror-edges. To expand $VP(q)$ by exactly $k$ units of area, any solution needs to contain a subset of these main-mirror-edges, and this subset exists if and only if there exists an equivalent subset in $InSS$. Thus, we can reduce this problem to the Subset-Sum problem in polynomial time.

The reduction also works if we look for the minimum number of mirror-edges. In that case, the useless $\seg{rd}$-edges and other unnecessary edges should not be chosen to be converted to mirrors. However, those main-mirror-edges corresponding to the solution of $InSS$ are required to be selected to be converted into mirrors. Based on the solution, if in a gadget the main mirror-edge is not selected the $\seg{rd}$- edge on the opposite side should be converted to mirror so that all the $SMV$ area become mirror-visible to $q$.  Note that as we look for the optimal solution, we need to select the best subset of the main mirror-edges. So, the optimal version of the problem is NP-hard.

\subsubsection{Multiple Reflections}

As Fig.~\ref{fig.9}(B) illustrates, every $SMV$ area has an edge ($\seg{cx}$) which is not directly visible to $q$. Consider a gadget, we call the $\seg{cx}$ segment the ``second-mirror-edge". We can easily construct a gray-quadrangle region so that it becomes entirely mirror-visible to both the main mirror-edge and the second-mirror-edge. That is because we consider the diffuse type of reflections. We only need to make sure that gray-quadrangular region is in front of or at least alongside the second-mirror-edge, and it is not behind the second-mirror-edge. That is because the second-mirror-edge cannot see its behind.

Consider two reflections, $q$ can see a gray-quadrangular region through one $\seg{rd}$-edge and its corresponding second-mirror-edge on the opposite side gadget. Therefore, considering more than one reflections, $q$ can see a gray-quadrangular region either by a main-mirror-edge, or a second-mirror-edge. However, if the main -mirror-edge is not selected, for the second-mirror-edge to work we need to convert a $\seg{rd}$-edge on the opposite side of that gadget too.

Now, we can use a reduction similar to the one used in the single reflection case, except that here we need to count on both the main-mirror-edge and the second-mirror-edge. We see if either one of these mirror-edges is selected, then we choose the corresponding value of $InSS$.

\subsubsection{Optimal Version of Multiple Reflections}

If we look for the minimum number of mirror-edges, the previous reduction works. That is because
the main-mirror-edge costs only one edge to be converted into a mirror, but the second-mirror-edge requires two mirror-edges in the middle.

 Here again, if a $SMV$ area is covered by a main-mirror-edge, there is no need to convert the corresponding $\seg{rd}$-edge to mirror.

No matter how many multiple reflections are allowed, $q$ cannot see any invisible gray-quadrangular region except by using either the main mirror-edge or the second-mirror-edge as reflectors. And, if we select either one of these mirror-edges, the entire corresponding invisible gray-quadrangular region gets mirror-visible to $q$ (remember that for the second-mirror-edge to works individually the $\seg{rd}$-edge) needs to be converted to mirror too).

\section{Discussion}

In this paper we have obtained complexity results for the problem of adding exactly k units of area to the visibility polygon of a point within a single polygon. Diffuse and Specular type of reflections as well as single or multiple reflections have been carefully considered and addressed. These are the very first result of its kind on this problem to the best of authors knowledge. We intend to investigate the problem of adding at least k units of area to the visibility polygon of a point within a simple polygon in further steps. Also, proposing an approximation algorithm for any version of the problem would be priceless.

The specular version considering multiple reflections is still open.
\linespread{1} 

\end{document}